\documentclass[12pt, draftclsnofoot, onecolumn]{IEEEtran}
\usepackage[utf8]{inputenc}
\usepackage{tabularx}
\usepackage{amsmath}
\usepackage{amsthm}
\usepackage{amssymb}
\usepackage{enumerate}
\usepackage{latexsym}
\usepackage{mathtools}
\usepackage{multirow}
\usepackage[ruled,linesnumbered]{algorithm2e}
\usepackage{longtable}

\usepackage{makecell}

\SetKwRepeat{Do}{do}{while}

\usepackage[unicode=true,
bookmarks=true,bookmarksnumbered=true,bookmarksopen=true,bookmarksopenlevel=1,
breaklinks=false,pdfborder={0 0 0},pdfborderstyle={},backref=false,colorlinks=false]
{hyperref}

\makeatletter
\theoremstyle{plain}

\newtheorem{theorem}{Theorem}[section]
\newtheorem{corollary}{Corollary}[section]

\newtheorem{lemma}{Lemma}[section]

\newtheorem{example}{Example}[section]
\theoremstyle{remark}

\newcommand{\F}{\mathbb{F}}

\usepackage{color, colortbl}
\usepackage[noadjust,sort,compress]{cite}

\@ifundefined{showcaptionsetup}{}{%
	\PassOptionsToPackage{caption=false}{subfig}}
\usepackage{subfig}

\usepackage[justification=centering]{caption}

\definecolor{Gray}{gray}{0.9}
\usepackage[first=0,last=9]{lcg}

\begin{document}

\title{Optimal Codes with Positive Griesmer Defects, Related Optimal and Almost Optimal LRC Codes}

\author{Yurui Wang, Hao Chen and Xia Wu\thanks{Yurui Wang and Xia Wu are with School of Mathematics, Southeast University , Nanjiang 210096, Jiangsu, China (e-mail: wangyurui0201@163.com, wuxia80@seu.edu.cn). Hao Chen is with College of Information Science and Technology, Jinan University, Guangzhou, Guangdong Province, 510632, China (e-mail: haochen@jnu.edu.cn). The research of Hao Chen was supported by NSFC Grant 62032009. The work of Xia Wu was supported by NSFC Grant 12371035 and by Jiangsu
Provincial Scientific Research Center of Applied Mathematics under Grant
BK20233002. (Corresponding author: Hao Chen)}}

\maketitle

\begin{abstract}
Solomon and Stiffler constructed infinitely many families of linear codes meeting the Griesmer bound in 1965. It is well-known in 1990's that certain Griesmer codes (codes with the zero Griesmer defect) are equivalent to Solomon-Stiffler codes or Belov codes. Griesmer codes constructed in some recent papers published in IEEE Trans. Inf. Theory are actually Solomon-Stiffler codes or affine Solomon-Stiffler codes proposed in our previous paper. Therefore it is more challenging to construct optimal codes with positive Griesmer defects.\\

In this paper, we construct several infinite families of optimal codes with positive Griesmer defects. Then these codes are certainly not equivalent to Solomon-Stiffler codes or Belov codes. Weight distributions and subcode support weight distributions of these optimal codes are determined. On the other hand, some of constructed optimal linear codes are optimal locally recoverable codes (LRCs) meeting the Cadambe-Mazumdar (CM) bound. Some of our constructed optimal codes are very close to the CM bound. Localities of these optimal or almost optimal LRC codes are two.\\


\end{abstract}

\section{Introduction}

\subsection{Preliminaries}

Let $q$ be a prime power and ${\bf F}_q$ be the finite field with $q$ elements. We denote by ${\bf F}_q^n$ the length $n$ linear space over ${\bf F}_q$. The linear space ${\bf F}_q^n$ is endowed with the Hamming metric. The support of a vector ${\bf a}=(a_0,\ldots,a_{n-1})\in {\bf F}_q^n$ is $$supp({\bf a})=\{i: a_i \neq 0\}.$$ The Hamming weight $wt({\bf a})$ of a vector ${\bf a} \in {\bf F}_q^n$ is the cardinality of its support. The distance $d({\bf a}, {\bf b})$ between two vectors ${\bf a}$ and ${\bf b}$ is $d({\bf a}, {\bf b})=wt({\bf a}-{\bf b})$. The minimum distance of a code ${\bf C} \subset {\bf F}_q^n$ is, $$d({\bf C})=\min_{{\bf a} \neq {\bf b}} \{d({\bf a}, {\bf b}),  {\bf a} \in {\bf C}, {\bf b} \in {\bf C} \}.$$  A linear $[n,k,d]_q$ code in ${\bf F}_q^n$ of a dimension $k$ subspace with the minimum distance $d$. It is obvious that the minimum distance of a linear code is the minimum weight of its nonzero codewords.  In coding theory, one of the main goals is to construct linear codes with large $k$ and $d$ when $n$ is given.  We call a linear $[n,k,d]_q$ code optimal, if there is no linear $[n,k,d+1]_q$ code. Let $${\bf C}^{\perp}=\{(y_1,\ldots,y_n) \in {\bf F}_q: \Sigma_{i=1}^n y_ix_i=0, \forall(x_1,\ldots,x_n) \in {\bf C}\}$$ be the dual code of the linear code ${\bf C} \subset {\bf F}_q^n$. The minimum
distance $d^{\perp}$ of ${\bf C}^{\perp}$ is called the dual distance of ${\bf C}$. The linear code ${\bf C}$ is called projective if its dual distance is at least $3$. We refer to \cite{HP,MScode} for the detail.

The Singleton bound in \cite{Singleton} for linear codes asserts
$$d \leq n-k+1$$ and codes attaining this bound are called maximal distance separable
(MDS) codes. Reed-Solomon codes are well-known linear MDS codes, see \cite{MScode}. The Griesmer bound for a linear $[n, k, d]_q$ code in \cite{Griesmer} asserts $$n \geq \Sigma_{i=0}^{k-1} \left\lceil \frac{d}{q^i} \right\rceil.$$
A linear $[n, k, d]_q$ code meeting this bound is called a Griesmer code. It is clear that Griesmer codes are optimal. We set $g_q(k,d)=\Sigma_{i=0}^{k-1} \left\lceil \frac{d}{q^i} \right\rceil$. The difference $GD=n-g_q(k,d)$ is called the Griesmer defect of this linear $[n,k,d]_q$ code. Linear code with small Griesmer defects are in general have good parameters. Set $\delta_q(k,d, d')=g_q(k,d+d')-g_q(k,d)$, $d'\geq 1$. If $\delta_q(k,d,d')$ is larger than the Griesmer defect, the optimal minimum weight of this linear $[n, k]_q$ code is at most $d+d'-1$.

\subsection{Subcode support weight distribution}

Let $A_i({\bf C})$ be the number of codewords $i$ in the code ${\bf C}$, then $$\{A_0({\bf C}), A_1({\bf C}), \ldots,
A_n({\bf C})\}$$ is called the weight distribution of this code. We refer to \cite{Wei} for the generalized Hamming weights of linear codes. Let ${\bf D} \subset {\bf C}$ be a nonzero subcode.
The support of ${\bf D}$ consists of the nonzero positions of all codewords in ${\bf D}$, that is,
$$supp({\bf D})=\{i:\exists \,{\bf a}=(a_{1}, a_2, \cdots,a_{n})\in {\bf D} \text{ with } a_{i}\neq0\}.$$
The weight of the subcode ${\bf D}$ is $|supp({\bf D})|$.
The $r$-generalized Hamming weight ($r$-GHW) of ${\bf C}$ is defined as
$$d_{r}({\bf C})=\min\{|supp({\bf D})|:\text{${\bf D}$ is a subcode of ${\bf C}$ with dimension $r$}\}.$$
The sequence $[d_{1}({\bf C}),d_{2}({\bf C}),\cdots,d_{k}({\bf C})$ is called the weight hierarchy of ${\bf C}$.
When $r=1$, the parameter $d_{1}({\bf C})$ is the minimum distance $d({\bf C})$.
The number of subcodes with dimension $r$ in a linear code ${\bf C} \in \F_q^n$ with dimension $k$ is given by $Gaussian(k,r)_q$, where $1\leq r \leq k$ and $Gaussian(k,r)_q$ represents the $q$-ary Gaussian binomial coefficient,  defined as
\begin{equation}\label{e.2}
    Gaussian(k,r)_q=\frac{(q^k-1)(q^k-q)\cdots(q^k-q^{r-1})}{(q^r-1)(q^r-q)\cdots(q^r-q^{r-1})}.
\end{equation}

The sequence $[A_1^{r}({\bf C}), A_2^{r}({\bf C}),\cdots, A_n^{r}({\bf C})]$ is called the $r$-subcode support weight distribution ($r$-SSWD) of ${\bf C}$,
where
$$A_i^{r}({\bf C})=|\{ {\bf D}:{\bf D} \text{ is a subspace of ${\bf C}$ with dimension $r$ and $|supp ({\bf D})|=i$} \}|,$$
for $0\leq i\leq n$ and $1\leq r\leq k.$
When $r=1$, it is clear that $A_i({\bf C}) = (q-1)A_i^1({\bf C})$. Then subcode support weight distributions are natural generalizations of weight distributions of linear codes.

Let ${\bf C} \subset {\bf F}_q^n$ be a linear $[n,k]_q$ code. ${\bf C}$ be a linear $[n,k]_q$ code with the dual distance at least $3$. Let $G$ be one generator matrix of ${\bf C}$. Then any two columns in $G$ are linearly independent vectors in ${\bf F}_q^k$. Suppose that $H$ takes all codimension $r$ linear subspaces of ${\bf F}_q^k$. It is well-known that $$d_r({\bf C})=\min_{H} \{n-|H \bigcap G|\},$$ where $H \bigcap G$ is the set of columns of $G$ lied in the codimension $r$ linear subspace $H$. Then $$A_i^{r}({\bf C})=|\{H: |H \bigcap G|=n-i\}|,$$ that is, $A_i^{r}({\bf C})$ is number of codimension $r$ linear subspace of ${\bf F}_q^k$ satisfying $|H \bigcap G|=n-i$.

For a $k$-dimensional vector space $U$ over $\mathbf{F}_q$, assume
\begin{equation}
    SUB(U)=\{V|V\ \text{is a nonzero subspace of } U \}
\end{equation}
and
\begin{equation}
    SUB^r(U)=\{V|V\ \text{is a  subspace of }\ U \ \text{with dimension}\ r\}.
\end{equation}
Let $G=[\alpha_1,\cdots,\alpha_n]$ be a $k\times n$ matrix, where $\alpha_i^T\in \mathbf{F}_q^k$ for $1\leq i\leq n$. For any $V\in SUB(\mathbf{F_q^k})$, the function $m_G:SUB(\mathbf{F_q^k})\rightarrow \mathbf{N}$ is defined as following:
\begin{equation}
    m_G(V)=|\{1\leq i\leq n |\alpha_i^T \in V\}|.
\end{equation}
Let $\mathbf{C}$ be an $[n,k]_q$-linear code with a generator matrix $G$. We have for any subspace $U\in SUB^r(\mathbf{C})$, there exists a subspace $V\in SUB^r(\mathbf{F}_q^k)$ such that $U=\{\mathbf{x}G|\mathbf{x}\in V\}$.
The following lemma is easy to verify.
\begin{lemma}
    For an $[n,k]_q$-linear code $\mathbf{C}$ and a subspace $U\in SUB^r(\mathbf{C})$, the subcode support weight of $U$ is
    \begin{equation}
        |supp(U)|=n-m_G(V^\perp).
    \end{equation}
    And the $r$-GHW of $\mathbf{C}$ is
    \begin{equation}
        d_r(\mathbf{C})=n-max\{m_G(V)|V\in SUB^{k-r}(\mathbf{F}_q^k)\}
    \end{equation}
    where $1\leq r\leq k.$
\end{lemma}

\subsection{Locally recoverable codes and the Cadambe-Mazudar bound}

For a linear code C over ${\bf F}_q$ with length $n$,
dimension $k$ and the minimum distance $d$, we define the locality as follows. Given $a \in {\bf F}_q$, set $$C(i,a)=\{{\bf x} \in {\bf C}: x_i=a\},$$ where $i \in \{1,2,\ldots,n\}$.  The
set $C_A(i,a)$ is the restriction of $C(i,a)$ to the set $A \subset \{1,2,\ldots,n\}$ of coordinate
positions. The linear code ${\bf C}$ is a locally
recoverable code (LRC) with the locality $r$, if for each $i \in \{1,2,\ldots,n\}$,
there exists a subset $A_i \subset \{1,2,\ldots,n\}\setminus \{i\}$ of cardinality at
most $r$, such that $$C_{A_i}(i,a)\bigcap C_{A_i}(i,a')=\emptyset$$ for any given $a \neq a'$, see \cite{Gopalan,PD}. It was proved in \cite{Gopalan,PD} the following Singleton-like bound for LRC codes, $$d \leq n-k+2-\left\lceil \frac{k}{r}\right\rceil.$$

Another important bound on LRC codes is the following Cadambe-Mazumdar (CM) bound proposed in \cite{CM}.\\

\begin{theorem}
    For any \(r\)-locally recoverable code \((n, k, d)\) over \(\mathbb{F}_q\),
\begin{equation}
    k \leq \min_{t \in \mathbb{Z}_+} \{ t r + k_{\text{opt}}^{(q)}(n - t(r + 1), d) \},
\end{equation}
where $k_{\text{opt}}^{(q)}(n,d)$ is the maximum dimension $k$ of linear $[n,k,d]_q$ codes.
\end{theorem}

We call the difference of $\min_{t \in \mathbb{Z}_+} \{ t r + k_{\text{opt}}^{(q)}(n - t(r + 1), d) \}-k$ the CM (Cadambe-Mzumdar) defect. It is clear that a LRC code with the zero CM defect is optimal. These LRC codes with CM defect equal to one or two are very close to optimal case. We call these LRC codes with one or two CM defect almost optimal LRC codes.
\subsection{Related works}

In 1965, Solomon and Stiffler proposed a method to construct linear codes meeting the Griesmer bound, and infinitely many infinite families of binary and $q$-ary Griesmer codes were constructed in \cite{Solomon}. It was then proved that for fixed $q$ and $k$, there are Griesmer codes for any given fixed large minimum distance $d$, see \cite{BM,Hill}. Then numerous Griesmer codes have been constructed from various methods, see \cite{Hill,Hell,Heng,LDT,Hyun,LL,Hu}. Optimal linear codes close to the Griesmer bound were also constructed in \cite{BJV,KM,Hyun,Hu}. For example, optimal codes in \cite{KM} have Griesmer defects one or two.

Linear codes with few nonzero weights have been studied by many authors, see \cite{DingTang,HengDing,LDT,Hu} and the references therein. For Solomon-Stiffler codes and Belov codes, we refer to \cite{Hell2,Solomon}. Some of them have only few weights and their weight distributions can be determined, see \cite{Hyun}.
From the following result due to \cite{Hamada} and \cite{Hell2,Hell3}, some Griesmer code with their minimum weights satisfying certain condition are actually equivalent to Solomon-Stiffler codes or Belov codes.

\begin{theorem}[see \cite{Hell2,Hell3,Hamada}]\label{T-1-3}%Theorem 3
1) Any Griesmer $[g(d,k,2),k,d]_2$ code for which $d \leq 2^{k-1}$ is either a Solomon-Stiffler code or a Belov code.

2) Let ${\bf C}$ be a linear Griesmer $[n,k,d]_q$ code, where $d=q^{k-1}-\Sigma_{i=1}^hq^{u_i-1}$ and $k>u_1>u_2>\cdots>u_h \geq 1$. Then ${\bf C}$ is $q$-ary Solomon-Stiffler code.\\
\end{theorem}

We refer to Theorem 5 of \cite{Hell3} for the extension of conclusion 2) in $q$-ary linear Griesmer code case. As indicated in \cite{Xie}, many Griesmer codes constructed in previous papers are actually Solomon-Stiffler codes. {\bf For example, let us indicate that Griesmer codes 10 and 11 in Table XII of a recent early access paper in IEEE Trans. Inf. Theory \cite{Wu2026} are special case of affine Solomon-Stiffler codes constructed in \cite{Chen}.} Therefore, it is more interesting and challenging to construct optimal codes with positive Griesmer defects.

 There have been many papers to construct LRC codes, which are optimal with the Singleton-like bound or the Cadambe-Mazumdar bound, we refer to \cite{BTV,CM,LL,HMM,MP} and references therein.

\subsection{Our contributions}

In this paper, we construct several families of optimal codes with positive Griesmer defects. Then the results in \cite{Hamada} and \cite{Hell2} and Theorem 1.2 cannot be applied to these optimal codes and they are not equivalent to known Griesmer codes. Weight distributions and subcode support weight distributions (SSWD) of these optimal codes are determined explicitly. Subcode support weight distributions of some optimal codes are also determined. More interestingly, these optimal codes are LRC codes with small locality $2$. Some of these LRC codes meet the CM bound (see Theorem 4.3). and they are optimal LRC codes simultaneously. Some of them are very close to the CM-bound (see Theorem 4.2 and 4.3). The method in this paper is inspired from the classical work of Solomon and Stiffler \cite{Solomon}.

The paper is organized as follows. In Section 2 and Section 3, we construct several families of optimal codes with positive Griesmer defects. Their weight distributions and subcode support weight distributions are determined. In Section 4, we prove that some of these optimal codes meet the CM bound simultaneously and some of them have CM defects equal to one or two. Section 5 concludes the paper.

\section{First family of optimal codes, their weight distributions and subcode support weight distributions}

Let $h$ be a positive integer satisfying $0< h\leq q^m-1$. The $q$-adic expansion of $h$ is defined by
\begin{equation}
    h=h_{m-1}q^{m-1}+\dots +h_0,
\end{equation}
where $0\leq h_i<q$ for any $0\leq i\leq m-1$.
Let $i_h$ be the minimum integer such that $h_{i_h}>0$ where $i_h\geq 0$.

\subsection{Construction}

Let $k$ be a positive integer satisfying $k\ge3$. Let $u_1,\dots,u_h$ be integers satisfying $u_1=\cdots=u_h=u\ge2$. Let $U_1,\dots,U_h$ be $h$ linear subspaces of $\mathbf{F}_q^k$ of dimension $u_1,\dots,u_h$ satisfying
\begin{equation}
    U_i \cap U_j = \{0\}, \   \text{if} \ i \neq j ;
\end{equation}
\begin{equation}\label{eq:rank1}
     q^k - q^{k-1} > h(q^u-1).
\end{equation}

Let $\mathbf{g}_1, \dots,\mathbf{g}_{q^k-1}$ be all nonzero vectors of  $\mathbf{F}_q^k$. Let $U$ be the set of all vectors among these $q^k-1$ nonzero vectors in  $\mathbf{F}_q^k$ by deleting these nonzero vectors in $U_1 \cup \cdots \cup U_h$. Then
\begin{equation}
    |U|=(q^k-1)-h(q^u-1).
\end{equation}
The column vectors in $U$ form a $k\times|U|$ matrix G.

If vector $\mathbf{v}$ is a column vector of $U$, $\lambda\mathbf{v}$ is also a column vector of $U$, where $\lambda \in \mathbf{F}_q^* $. If we pick up only one vector from $\{\lambda\mathbf{v}|\lambda \in \mathbf{F}_q^* \}$ for every vector $\mathbf{v}\in U$, these vectors form a $k\times \frac{|U|}{q-1}$ matrix $\widetilde{G}$.

\begin{lemma}\label{lemma2.1}
    The rank of $\widetilde{G}$ is k.
\end{lemma}
    \begin{proof}
        We have $U$ is not contained in any $(k-1)$-dimensional subspace of $\mathbf{F}_q^k$, since
        $|U|=(q^k-1)-h(q^u-1)>q^{k-1}-1$ with (\ref{eq:rank1}). Then $r(\widetilde{G})=r(G)=k$.
    \end{proof}

From \textbf{Lemma} \ref{lemma2.1}, let $\mathbf{C}$ be the code with the generator matrix $\widetilde{G}$. Then we have the following result.

\begin{theorem}\label{theorem2.1}
        The code $\mathbf{C}$ is a linear $[\frac{(q^k-1)-h(q^u-1)]}{q-1},k,d\geq q^{k-1}-hq^{u-1}]_q$ code. This code has at most $h+1$ weights.
\end{theorem}
    \begin{proof}
    For any hyperplane $H\subset \mathbf{F}_q^k$, let $H$ be defined by
    \begin{equation*}
        a_1x_1+\dots +a_kx_k=0
    \end{equation*}
    where $\mathbf{a}=(a_1,\dots,a_k)$. The weight of nonzero codeword $\mathbf{a}\cdot \widetilde{G}$ corresponding to $H$ is
    \begin{equation}
        \frac{|U|-|H-H\cap (U_1\cup\dots\cup U_h)|}{q-1}\geq q^{k-1}-hq^{u-1}.
    \end{equation}
    Since the dimension of $H\cap U_i$ has two possible values $u$ and $u-1$, then the conclusion follows directly.
    \end{proof}

\begin{lemma}\label{lemma2.2}
    The Griesmer defect of $\mathbf{C}$ is lower than $\sum_{i=1}^{k-u}\lfloor\frac{h}{q^i}\rfloor$.
\end{lemma}
    \begin{proof}
    \begin{equation*}
    \begin{aligned}
        GD&=n-g_q(k,d)\\
        &\leq \frac{(q^k-1)-h(q^u-1)}{q-1}-\sum_{i=0}^{k-1} \left\lceil \frac{q^{k-1}-hq^{u-1}}{q^i} \right\rceil\\
        &=\sum_{i=1}^{k-u}\lfloor\frac{h}{q_i}\rfloor.
    \end{aligned}
    \end{equation*}
    \end{proof}

Specially, the Griesmer defect of $\mathbf{C}$ is $\sum_{i=1}^{k-u}\lfloor\frac{h}{q^i}\rfloor$ when $d=q^{k-1}-hq^{u-1}$.
\begin{theorem}\label{theorem2.2}

    The linear code $\mathbf{C}$ is a linear $[\frac{(q^k-1)-h(q^u-1)}{q-1},k,d=q^{k-1}-hq^{u-1}]_q$ code if $h\leq q^u$.
\end{theorem}
    \begin{proof}
    Otherwise, we have $d>q^{k-1}-hq^{u-1}$. For any hyperplane $H\subset \mathbf{F}_q^k$, let $H$ be defined by
    \begin{equation*}
        a_1x_1+\dots +a_kx_k=0
    \end{equation*}
    where $\mathbf{a}=(a_1,\dots,a_k)$. From the proof  of \textbf{Theorem} \ref{theorem2.1}, it can be obtained that there exists index $i$ such that $U_i\subseteq H$ for any hyperplanes. Then $\mathbf{F}_q^k$ is contained by $\bigcup_{i=1}^h U_i^\perp$. On the other hand,
    \begin{equation*}
        \begin{aligned}
         |\bigcup_{i=1}^h U_i^\perp|&\leq\sum_{i=1}^{h} |U_i^\perp|-(h-1)\\
         &=hq^{k-u}-(h-1)\\
         &<q^k.
        \end{aligned}
    \end{equation*}
    This leads to the contradiction.
    \end{proof}

\begin{theorem}\label{theorem2.3}
     The linear code $\mathbf{C}$ with parameters $[\frac{(q^k-1)-h(q^u-1)]}{q-1},k,d=q^{k-1}-hq^{u-1}]_q$ is distance optimal if $i_h+u>\sum_{i=1}^{k-u}\lfloor\frac{h}{q^i}\rfloor$.
\end{theorem}
    \begin{proof}
    Let $d_0=q^{k-1}-hq^{u-1}$.
    We have $\delta_q(k,d_0, d')=g_q(k,d_0+d')-g_q(k,d_0)=d'-1+i_h +u$ where $1\leq d'<q$. Then let $d'=1$ and we have $\delta_q(k,d_0, 1)=i_h +u$.
    So
    \begin{equation*}
    \begin{aligned}
     \sum_{i=0}^{k-1} \left\lceil \frac{d+1}{q^i}\right\rceil &\geq\sum_{i=0}^{k-1} \left\lceil \frac{d_0+1}{q^i}\right\rceil\\
     &=\sum_{i=0}^{k-1} \left\lceil \frac{d_0}{q^i}\right\rceil+i_h+u\\
     &=n-\sum_{i=1}^{k-u}\lfloor\frac{h}{q^i}\rfloor+i_h+u\\
     &>n.
     \end{aligned}
\end{equation*}
    We have the optimal distance of $[n, k]_q$ is at most $d=d_0$.
    \end{proof}

\begin{corollary}\label{corollary2.1}
    Let $q\geq u\geq2$. The linear code $\mathbf{C}$ with parameters $[\frac{(q^k-1)-h(q^u-1)]}{q-1},k,d=q^{k-1}-hq^{u-1}]_q$ is distance optimal when $h\leq uq$. In particular if $u=2$, the linear code $\mathbf{C}$ is distance optimal when $h\leq 2q$.
\end{corollary}

\begin{proof}
\begin{enumerate}
    \item[(a)] If $h<qu$, we have $i_h+u\geq u > u-1\geq\sum_{i=1}^{k-u}\lfloor\frac{h}{q^i}\rfloor$.
    \item[(b)] If $h=qu,u<q$, we have $i_h+u= u+1 > u=\sum_{i=1}^{k-u}\lfloor\frac{h}{q^i}\rfloor$.
    \item[(c)] If $h=qu,u=q$, we have $i_h+u= u+2 > u+1=\sum_{i=1}^{k-u}\lfloor\frac{h}{q^i}\rfloor$.
\end{enumerate}
    According to \textbf{Theorem} \ref{theorem2.3}, the conclusion is proved.
\end{proof}

In the following table, we list some codes constructed in Corollary 2.1.
\begin{center}
\begin{tabular}{|c|c|c|c|}
\hline
$q$&$[n,k,d]$&$(k,u,h)$&Optimality\\
\hline
$2$&$[22,5,10]$&$(5,2,3)$&optimal\\
\hline
$2$&$[19,5,8]$&$(5,2,4)$&optimal\\
\hline
$2$&$[54,6,26]$&$(6,2,3)$&optimal\\
\hline
$2$&$[51,6,24]$&$(6,2,4)$&optimal\\
\hline
$2$&$[118,7,58]$&$(7,2,3)$&optimal\\
\hline
$2$&$[115,7,56]$&$(7,2,4)$&optimal\\
\hline
$2$&$[246,8,122]$&$(8,2,3)$&optimal\\
\hline
$2$&$[243,8,120]$&$(8,2,4)$&optimal\\
\hline
$3$&$[24,4,15]$&$(4,2,4)$&optimal\\
\hline
$3$&$[20,4,12]$&$(4,2,5)$&optimal\\
\hline
$3$&$[16,4,9]$&$(4,2,6)$&optimal\\
\hline
$3$&$[105,5,69]$&$(5,2,4)$&optimal\\
\hline
$3$&$[101,5,66]$&$(5,2,5)$&optimal\\
\hline
$3$&$[97,5,63]$&$(5,2,6)$&optimal\\
\hline
\end{tabular}
\end{center}

\subsection{Weight distribution and subcode support weight distribution}

In this subsection, we determine weight distribution and subcode support weight distribution of these codes constructed in the previous subsection.

\begin{theorem}\label{theorem2.4}
 Assume the notations are as above. For integers $u,h$ satisfying $u\geq2, i_h+u>\sum_{i=1}^{k-u}\lfloor\frac{h}{q^i}\rfloor, k\geq hu$  and $q^k - q^{k-1} > h(q^u-1)$, we construct an optimal linear $[n,k,d]_q$ code $\mathbf{C}_1$ with
    \begin{equation*}
        n=\frac{(q^k-1)-h(q^u-1)}{q-1}
    \end{equation*}
    and
    \begin{equation*}
        d=q^{k-1}-hq^{u-1}.
    \end{equation*}

\begin{enumerate}
\item[(a)]
    The weight distribution of the linear code $\mathbf{C}_1$ is determined in the following table.
\begin{center}
\begin{tabular}{|c|c|}
\hline
weight & multiplicity \\
\hline
0 & 1 \\
\hline
$q^{k-1}-hq^{u-1}$ & $(q^u-1)^hq^{k-hu}$ \\
\hline
$q^{k-1}-(h-1)q^{u-1}$ & $h(q^u-1)^{h-1}q^{k-hu}$ \\
\hline
$\vdots$ & $\vdots$ \\
\hline
$q^{k-1}-q^{u-1}$ & $h(q^u-1)q^{k-hu}$ \\
\hline
$q^{k-1}$ & $q^{k-hu}-1$ \\
\hline
\end{tabular}
\end{center}

\item [(b)]
The $r$-GHW of $\mathbf{C}_1$ is as follows,
\begin{equation}
  d_r(\mathbf{C})= \begin{cases}
      \frac{(q^k-q^{k-r})-h(q^u-q^{u-r})}{q-1} ,1\leq r\leq u\\
      \frac{(q^k-q^{k-r})-h(q^u-1)}{q-1},r\geq u
  \end{cases}
\end{equation}

\end{enumerate}
\end{theorem}
\begin{proof}
 Let $\mathbf{e}_i$ be the vector with all $0$s coordinates except one $1$ coordinate at the $i$th position.
    Assume $U_i$ is the linear space generated by $\{\mathbf{e}_{(i-1)u+1},\dots,\mathbf{e}_{iu}\}$ when $i=1,\dots, h$. Let $U$ be the set of all vectors among these $q^k-1$ nonzero vectors in  $\mathbf{F}_q^k$ by deleting these nonzero vectors in $U_1 \cup \cdots \cup U_h$.
    If vector $\mathbf{v}$ is a column vector of $U$, $\lambda\mathbf{v}$ is also a column vector of $U$ where $\lambda \in \mathbf{F}_q^* $. If we pick up only one vector from $\{\lambda\mathbf{v}|\lambda \in \mathbf{F}_q^* \}$ for every vector $\mathbf{v}\in U$, these vectors form a $k\times \frac{|U|}{q-1}$ matrix $\widetilde{G}$.  According to \textbf{Theorem} \ref{theorem2.3}, we have $\mathbf{C}_1$ is optimal with
    \begin{equation*}
        n=\frac{(q^k-1)-h(q^u-1)}{q-1}
    \end{equation*}
    and
    \begin{equation*}
        d=q^{k-1}-hq^{u-1}.
    \end{equation*}
\begin{enumerate}
    \item[(a)]
    Let $H$ be a hyperplane of $\mathbf{F}_q^k$.
    Set $x=|\{i|i\in\{1,\dots ,q\}, dim(H\cap U_i)=u\}|$. We proceed to prove the conclusion in the following two cases listed in the following table.
    \begin{center}
\begin{tabular}{|c|c|c|c|}
\hline
&condition of $x$&weight & multiplicity \\
\hline
&&0 & 1 \\
\hline
case \romannumeral1&$x=h$&$q^{k-1}$ & $q^{k-hu}-1$ \\
\hline
case \romannumeral2&$x\neq h$&$q^{k-1}-(h-x)q^{u-1}$ & $\binom{h}{x}(q^u-1)^{h-x}q^{k-hu}$ \\
\hline

\end{tabular}
     \end{center}
\begin{enumerate}[(i)]
    \item It is clear that $dim(H\cap U_i)=u$ for each $i\in \{1,\dots,h\}$.\\
    Then we have
    \begin{equation*}
    \begin{aligned}
        wt(\mathbf{c})&=\frac{|U|-|H-H\cap (U_1\cup\dots\cup U_h)|}{q-1}\\
        &=q^{k-1}.
    \end{aligned}
    \end{equation*}
 And we know that
     \begin{equation}
        |\{\mathbf{a}|\mathbf{a}\neq\mathbf{0},U_i\subseteq H,i=1,\dots,h\}|=q^{k-hu}-1.
    \end{equation}
    The rest cases are similar.
\end{enumerate}
\item [(b)]
If $1\leq r\leq u$, we have
\begin{equation*}
    \begin{aligned}
        d_r(\mathbf{C}_1)&=n-max\{m_{\widetilde{G}}(V)|V\in SUB^{k-r}(\mathbf{F}_q^k)\}\\
        &=n-\frac{max\{m_{G}(V)|V\in SUB^{k-r}(\mathbf{F}_q^k)\}}{q-1}\\
        &=n-\frac{max\{(|V|-1)-\sum_{i=1}^{h}(|V\cap U_i|-1)|V\in SUB^{k-r}(\mathbf{F}_q^k)\}}{q-1}\\
        &\geq \frac{(q^k-q^{k-r})-h(q^u-q^{u-r})}{q-1}.
    \end{aligned}
\end{equation*}
Let $\boldsymbol{\alpha}_{(i-1)(u-r)+1}=\mathbf{e}_{(i-1)u+1},\dots,\boldsymbol{\alpha}_{i(u-r)}=\mathbf{e}_{iu-r}$ for $i=1,\dots,h$, and $\boldsymbol{\alpha}_{h(u-r)+1}=\mathbf{e}_{u-r+1}+\mathbf{e}_{2u-r+1},\dots,\boldsymbol{\alpha}_{h(u-r)+r}=\mathbf{e}_u+\mathbf{e}_{2u},\boldsymbol{\alpha}_{h(u-r)+r+1}=\mathbf{e}_{u-r+1}+\mathbf{e}_{3u-r+1},\dots,\boldsymbol{\alpha}_{h(u-r)+2r}=\mathbf{e}_u+\mathbf{e}_{3u},\dots,\boldsymbol{\alpha}_{h(u-r)+(h-2)r+1}=\mathbf{e}_{u-r+1}+\mathbf{e}_{hu-r+1},\dots,\boldsymbol{\alpha}_{h(u-r)+(h-1)r}=\mathbf{e}_u+\mathbf{e}_{hu},\boldsymbol{\alpha}_{hu-r+1}=\mathbf{e}_{hu+1},\dots, \boldsymbol{\alpha}_{k-r}=\mathbf{e}_{k}$. Let $V$ be the subspace generated by $\boldsymbol{\alpha}_1,\dots,\boldsymbol{\alpha}_{k-r}$. We have
\begin{equation}
    d_r(\mathbf{C}_1)=\frac{(q^k-q^{k-r})-h(q^u-q^{u-r})}{q-1}, 1\leq r\leq u .
\end{equation}
The rest cases are similar.

\end{enumerate}

\end{proof}

\begin{example}
    In the construction of Theorem 2.4, let $q=2,k=8,h=4$ and $u=2$. We construct an optimal linear $[243,8,120]_2$ code with the weight distribution as in the following table.
\begin{center}
\begin{tabular}{|c|c|}
\hline
weight & multiplicity \\
\hline
0 & 1 \\
\hline
120 & 81\\
\hline
122 & 108 \\
\hline
124 & 54 \\
\hline
126 & 12 \\
\hline

\end{tabular}
\end{center}

\end{example}

\begin{theorem}\label{theorem2.5}
    Assume the notation is as given above. For integers $u$ and $h$ satisfying $u=2,$ $h=2q$ and $k\geq 8$, we construct an optimal linear $[n,k,d]_q$ code $\mathbf{C}_2$ with
    \begin{equation*}
        n=\frac{(q^k-1)-2q(q^2-1)}{q-1}
    \end{equation*}
    and
    \begin{equation*}
        d=q^{k-1}-2q^2.
    \end{equation*}

\begin{enumerate}
\item[(a)]
    The weight distribution of the linear code $\mathbf{C}_2$ is determined in the following table.
\begin{center}
\begin{tabular}{|c|c|}
\hline
weight & multiplicity \\
\hline
0 & 1 \\
\hline
$q^{k-1}$ & $q^{k-8}-1$ \\
\hline
$q^{k-1}-q^2+q$ & $2q(q^2-1)q^{k-8}$ \\
\hline
$q^{k-1}-q^2$ & $2(q^{4}-q^{3}+q^{1}-1)q^{k-8}$ \\
\hline
$q^{k-1}-2q^2+2q$ & $(q^2-1)^2q^{k-6}$ \\
\hline
$q^{k-1}-2q^2+q$ & $2q(q^2-1)(q^{4}-q^{3}+q^{1}-1)q^{k-8}$ \\
\hline
$q^{k-1}-2q^2$ & $(q^{4}-q^{3}+q^{1}-1)^2q^{k-8}$ \\
\hline
\end{tabular}
\end{center}

\item [(b)]
The $r$-GHW of $\mathbf{C}_2$ is determined as follows,
\begin{equation}
  d_r(\mathbf{C})= \begin{cases}
      q^{k-1}-2q^2, \ r=1\\
      \frac{(q^k-q^{k-r})-2q(q^2-1)}{q-1}, \ r\geq2
  \end{cases}
\end{equation}

\end{enumerate}
\end{theorem}
\begin{proof}
    Let $f_1=0,f_2,\dots,f_{q-1},f_{q}=1$ be the elements of $\mathbf{F}_q$. Let $\mathbf{e}_i$ be the vector with all $0$s coordinates except one $1$ coordinate at the $i$th position.\\
    Assume $U_i$ is the linear space generated by $\{\mathbf{e}_1+f_i\mathbf{e}_2,\mathbf{e}_3+f_i\mathbf{e}_4\}$ when $i=1,\dots , q$ and $U_i$ is the linear space generated by $\{\mathbf{e}_5+f_{i-q}\mathbf{e}_6,\mathbf{e}_7+f_{i-q}\mathbf{e}_8\}$ when $i=q+1,\dots , 2q$. Let $U$ be the set of all vectors among these $q^k-1$ nonzero vectors in  $\mathbf{F}_q^k$ by deleting these nonzero vectors in $U_1 \cup \cdots \cup U_h$.
    If vector $\mathbf{v}$ is a column vector of $U$, $\lambda\mathbf{v}$ is also a column vector of $U$  where $\lambda \in \mathbf{F}_q^* $. If we pick up only one vector from $\{\lambda\mathbf{v}|\lambda \in \mathbf{F}_q^* \}$ for every vector $\mathbf{v}\in U$, these vectors form a $k\times \frac{|U|}{q-1}$ matrix $\widetilde{G}$.  According to \textbf{Corollary} \ref{corollary2.1}, $\mathbf{C}_2$ is optimal with
    \begin{equation*}
        n=\frac{(q^k-1)-2q(q^2-1)}{q-1}
    \end{equation*}
    and
    \begin{equation*}
        d=q^{k-1}-2q^2.
    \end{equation*}
\begin{enumerate}
    \item[(a)]
    Let $H$ be a hyperplane of $\mathbf{F}_q^k$. The linear space generated by $\{\mathbf{e}_1,\mathbf{e}_2,\mathbf{e}_3,\mathbf{e}_4\}$ is contained by $H$ if $U_i,U_j$ is contained by $H$ when $i,j\in\{1,\dots,q\}$ and $i\neq j$ since $U_i+U_j$ is the linear space generated by $\{\mathbf{e}_1,\mathbf{e}_2,\mathbf{e}_3,\mathbf{e}_4\}$. Then we have $U_i^\perp \cap U_j^\perp$ is the linear space generated by $\{\mathbf{e}_5,\mathbf{e}_6,\mathbf{e}_7,\mathbf{e}_8,\dots,\mathbf{e}_k\}$ when $i,j\in\{1,\dots,q\}$ and $i\neq j$. It is similar when $i,j\in\{q+1,\dots,2q\}$ and $i\neq j$.\\
    Set $$x_1=|\{i|i\in \{1,\dots ,q\}, dim(H\cap U_i)=2\}|$$ and $$x_2=|\{i|i\in \{q+1,\dots ,2q\}, dim(H\cap U_i)=2\}|.$$ The six cases are listed in the following table.
    \begin{center}
\begin{tabular}{|c|c|c|c|}
\hline
&condition of $x_1,x_2$&weight & multiplicity \\
\hline
&&0 & 1 \\
\hline
case \romannumeral1&$x_1=q,x_2=q$&$q^{k-1}$ & $q^{k-8}-1$ \\
\hline
case \romannumeral2&\makecell{$x_1=1,x_2=q$ \\or $x_1=q,x_2=1$}&$q^{k-1}-q^2+q$ & $2q(q^2-1)q^{k-8}$ \\
\hline
case \romannumeral3&\makecell{$x_1=0,x_2=q$ \\or $x_1=q,x_2=0$}&$q^{k-1}-q^2$ & $2(q^{4}-q^{3}+q^{1}-1)q^{k-8}$ \\
\hline
case \romannumeral4&$x_1=1,x_2=1$&$q^{k-1}-2q^2+2q$ & $(q^2-1)^2q^{k-6}$ \\
\hline
case \romannumeral5&\makecell{$x_1=0,x_2=1$\\or $x_1=1,x_2=0$}&$q^{k-1}-2q^2+q$ & $2q(q^2-1)(q^{4}-q^{3}+q^{1}-1)q^{k-8}$ \\
\hline
case \romannumeral6&$x_1=0,x_2=0$&$q^{k-1}-2q^2$ & $(q^{4}-q^{3}+q^{1}-1)^2q^{k-8}$ \\
\hline
\end{tabular}
     \end{center}
\begin{enumerate}[(i)]
    \item It is clear that $dim(H\cap U_i)=2$ for each $i\in \{1,\dots,2q\}$.\\
    Then we have
    \begin{equation*}
    \begin{aligned}
        wt(\mathbf{c})&=\frac{|U|-|H-H\cap (U_1\cup\dots\cup U_h)|}{q-1}\\
        &=q^{k-1}.
    \end{aligned}
    \end{equation*}
    And we know that
     \begin{equation}
        |\{\mathbf{a}|\mathbf{a}\neq\mathbf{0},U_i\subseteq H,i=1,\dots,2q\}|=q^{k-8}-1.
    \end{equation}
    The rest five cases are similar.
\end{enumerate}
\item [(b)]
If $r=1$, we have
\begin{equation*}
    d_1(\mathbf{C}_2)=d=q^{k-1}-2q^2.
\end{equation*}
If $r\geq2$, we have
\begin{equation*}
    \begin{aligned}
       d_r(\mathbf{C}_2)&=n-max\{m_{\widetilde{G}}(V)|V\in SUB^{k-r}(\mathbf{F}_q^k)\}\\
        &=n-\frac{max\{m_{G}(V)|V\in SUB^{k-r}(\mathbf{F}_q^k)\}}{q-1}\\
        &=n-\frac{max\{(|V|-1)-\sum_{i=1}^{2q}(|V\cap U_i|-1)|V\in SUB^{k-r}(\mathbf{F}_q^k)\}}{q-1}\\
        &\geq n-\frac{|V|-1}{q-1}\\
        &=\frac{(q^k-q^{k-r})-2q(q^2-1)}{q-1}.
    \end{aligned}
\end{equation*}
Let $\boldsymbol{\alpha}_1=\mathbf{e}_2,\boldsymbol{\alpha}_2=\mathbf{e}_4,\boldsymbol{\alpha}_3=\mathbf{e}_6,\boldsymbol{\alpha}_4=\mathbf{e}_8,\boldsymbol{\alpha}_5=\mathbf{e}_1+\mathbf{e}_5, \boldsymbol{\alpha}_6=\mathbf{e}_3+\mathbf{e}_7,\boldsymbol{\alpha}_7=\mathbf{e}_9,\dots , \boldsymbol{\alpha}_{k-2}=\mathbf{e}_k$ and $V$ be the subspace generated by $\boldsymbol{\alpha}_1,\dots,\boldsymbol{\alpha}_{k-r}$ for $r\geq 2$. Hence we have
\begin{equation}
    d_r(\mathbf{C}_2)=\frac{(q^k-q^{k-r})-2q(q^2-1)}{q-1}, \ r\geq 2.
\end{equation}

\end{enumerate}

\end{proof}

\begin{example}
    In the construction of Theorem 2.5, let $q=2,k=8,h=4$ and $u=2$. We construct an optimal linear $[243,8,120]_2$ code with the weight distribution as in the following table.
\begin{center}
\begin{tabular}{|c|c|}
\hline
weight & multiplicity \\
\hline
0 & 1 \\
\hline
120 & 81\\
\hline
122 & 108 \\
\hline
124 & 54 \\
\hline
126 & 12 \\
\hline

\end{tabular}
\end{center}

\end{example}

\section{Second family of optimal codes, their weight distributions and subcode support weight distributions}

\subsection{Construction}

Let $k$ be a positive integer satisfying $k\ge3$ and $u_1,\dots,u_h$ be integers satisfying $1\leq u_0<u_1\leq\cdots\leq u_h$ where $h\geq2$. Let $U_0,U_1,\dots,U_h$ be $h+1$ linear subspaces of $\mathbf{F}_q^k$ of dimension $u_0,u_1,\dots,u_h$ satisfying
\begin{equation}
    U_i \cap U_j = U_0,\  \text{if} \ i \neq j;
\end{equation}
\begin{equation}\label{eq:rank2}
     q^k - q^{k-1} > q^{u_0}-1+\sum_{i=1}^h (q^{u_i}-q^{u_0}).
\end{equation}

Let $\mathbf{g}_1, \dots, \mathbf{g}_{q^k-1}$ be all nonzero vectors of  $\mathbf{F}_q^k$. Let $U$ be the set of all vectors among these $q^k-1$ nonzero vectors in  $\mathbf{F}_q^k$ by deleting these nonzero vectors in $U_1 \cup \cdots \cup U_h$. Then we have
\begin{equation}
    |U|=(q^k-q^{u_0})-\sum_{i=1}^h (q^{u_i}-q^{u_0}).
\end{equation}
The column vectors in $U$ form a $k\times|U|$ matrix G.

If vector $\mathbf{v}$ is a column vector of $U$, $\lambda\mathbf{v}$ is also a column vector of $U$ where $\lambda \in \mathbf{F}_q^* $. If we pick up only one vector from $\{\lambda\mathbf{v}|\lambda \in \mathbf{F}_q^* \}$ for every vector $\mathbf{v}\in U$, these vectors form a $k\times \frac{|U|}{q-1}$ matrix $\widetilde{G}$.

\begin{lemma}\label{lemma3.1}
    The rank of $\widetilde{G}$ is k.
\end{lemma}
    \begin{proof}
        $U$ is not contained in any $(k-1)$-dimensional subspace of $\mathbf{F}_q^k$, since
        $|U|=(q^k-q^{u_0})-\sum_{i=1}^h (q^{u_i}-q^{u_0})>q^{k-1}-1$ with (\ref{eq:rank2}). Then $r(\widetilde{G})=r(G)=k$.
    \end{proof}

From \textbf{Lemma} \ref{lemma3.1}, let $\mathbf{C}$ be the code with the generator matrix $\widetilde{G}$. Then we have the following result.

\begin{theorem}\label{theorem3.1}
        The above code $\mathbf{C}$ is a linear $[\frac{(q^k-q^{u_0})-\sum_{i=1}^h (q^{u_i}-q^{u_0})}{q-1},k,d\geq q^{k-1}-\sum_{i=1}^h (q^{u_i-1})]_q$ code.
\end{theorem}
    \begin{proof}
    For any hyperplane $H\subset \mathbf{F}_q^k$, let $H$ be defined by
    \begin{equation*}
        a_1x_1+\dots +a_kx_k=0
    \end{equation*}
    where $\mathbf{a}=(a_1,\dots,a_k)$. The weight of nonzero codeword $\mathbf{a}\cdot \widetilde{G}$ corresponding to $H$ is
    \begin{equation*}
    \begin{aligned}
        &\frac{|U|-|H-H\cap (U_1\cup\dots\cup U_h)|}{q-1}\\
        =&\frac{|U|-|H|+|(H\cap U_0)\cup(H\cap (U_1-U_0))\cup\dots(H\cap (U_h-U_0))|}{q-1}\\
        =&\frac{|U|-|H|+|H\cap U_1|+\dots+|H\cap U_h|-(h-1)|H\cap U_0|}{q-1}\\
        \geq&q^{k-1}-\sum_{i=1}^h (q^{u_i-1}).
    \end{aligned}
    \end{equation*}
    \end{proof}

\begin{theorem}\label{theorem3.2}
    The code $\mathbf{C}$ is a linear $[\frac{(q^k-q^{u_0})-\sum_{i=1}^h (q^{u_i}-q^{u_0})}{q-1},k,d= q^{k-1}-\sum_{i=1}^h (q^{u_i-1})]_q$ code if $h\leq q^{u_1-u_0}$.
\end{theorem}
    \begin{proof}
    Otherwise, we have $d>q^{k-1}-hq^{u-1}$. For any hyperplane $H\subset \mathbf{F}_q^k$, let $H$ be defined by
    \begin{equation*}
        a_1x_1+\dots +a_kx_k=0
    \end{equation*}
    where $\mathbf{a}=(a_1,\dots,a_k)$. From the proof  of \textbf{Theorem} \ref{theorem2.1}, it can be obtained that there exists $i$ such that $U_i\subseteq H$ when $U_0\subseteq H$. Then we have $U_0^\perp$ is contained by $\bigcup_{i=1}^h U_i^\perp$. On the other hand,
    \begin{equation*}
        \begin{aligned}
         |\bigcup_{i=1}^h U_i^\perp|&\leq\sum_{i=1}^{h} |U_i^\perp|-(h-1)\\
         &\leq hq^{k-u_1}-(h-1)\\
         &<q^{k-u_0}.
        \end{aligned}
    \end{equation*}
    This leads to the contradiction.
    \end{proof}
\begin{lemma}\label{lemma3.2}
    The Griesmer defect of $\mathbf{C}$ can be upper bounded in the following several cases.
    \begin{enumerate}[1)]
        \item If $u_1=\dots =u_{s_1}<u_{s_1+1}=\dots =u_{s_1+s_2}<\dots <u_{s_1+\dots +s_{t-1}+1}=\dots =u_{s_1+\dots +s_{t}} $ where $1\leq s_i<q,i=1,\dots, t $ and $h=s_1+\dots +s_t$, the Griesmer defect of $\mathbf{C}$ is smaller than $(h-1)\frac{q^{u_0}-1}{q-1}$.
        \item If $u_1=\dots =u_h=u$, the Griesmer defect of $\mathbf{C}$ is smaller than $(h-1)\frac{q^{u_0}-1}{q-1}+\sum_{i=1}^{k-u}\lfloor\frac{h}{q^i}\rfloor.$
    \end{enumerate}
\end{lemma}
    \begin{proof}
    Let $d_0=q^{k-1}-\sum_{i=1}^h (q^{u_i-1})$.
        \begin{equation*}
        \begin{aligned}
        1)\ GD&\leq n-g_q(k,d_0)\\
        &= \frac{(q^k-q^{u_0})-\sum_{i=1}^h (q^{u_i}-q^{u_0})}{q-1}-\sum_{i=0}^{k-1} \left\lceil \frac{q^{k-1}-\sum_{i=1}^h (q^{u_i-1})}{q^i} \right\rceil\\
        &= \frac{(q^k-q^{u_0})-\sum_{i=1}^t s_i(q^{u_{\sum_{j=1}^i s_j}}-q^{u_0})}{q-1}-\sum_{i=0}^{k-1} \left\lceil \frac{q^{k-1}-\sum_{i=1}^t s_iq^{u_{\sum_{j=1}^i s_j}-1}}{q^i} \right\rceil\\
        &=(h-1)\frac{q^{u_0}-1}{q-1}.\\
        2)\ GD&\leq n-g_q(k,d_0)\\
        &= \frac{(q^k-q^{u_0})-h (q^u-q^{u_0})}{q-1}-\sum_{i=0}^{k-1} \left\lceil \frac{q^{k-1}-hq^{u-1}}{q^i} \right\rceil\\
        &=(h-1)\frac{q^{u_0}-1}{q-1}+\sum_{i=1}^{k-u}\lfloor\frac{h}{q^i}\rfloor.
        \end{aligned}
        \end{equation*}
    \end{proof}

\begin{theorem}\label{theorem3.3}
    The linear $[\frac{(q^k-q^{u_0})-\sum_{i=1}^h (q^{u_i}-q^{u_0})}{q-1},k,d=q^{k-1}-\sum_{i=1}^h (q^{u_i-1})]_q$ code $\mathbf{C}$ is distance optimal in the  following cases.
    \begin{enumerate}[1)]
    \item If $u_1=\dots =u_{s_1}<u_{s_1+1}=\dots =u_{s_1+s_2}<\dots <u_{s_1+\dots +s_{t-1}+1}=\dots =u_{s_1+\dots +s_{t}} $ where $1\leq s_i<q,i=1,\dots, t $ and $h=s_1+\dots +s_t$, the linear code $\mathbf{C}$ is distance optimal when $u_1>(h-1)\frac{q^{u_0}-1}{q-1}$.
    \item If $u_1=\dots =u_h=u$, the linear code $\mathbf{C}$ is distance optimal when $i_h+u>(h-1)\frac{q^{u_0}-1}{q-1}+\sum_{i=1}^{k-u}\lfloor\frac{h}{q^i}\rfloor$.
    \end{enumerate}
\end{theorem}
    \begin{proof}
Let $d_0=q^{k-1}-\sum_{i=1}^h (q^{u_i-1})$.\\
For $1)$, we have $\delta_q(k,d_0, 1)=g_q(k,d_0+1)-g_q(k,d_0)=u_1.$
    So
    \begin{equation*}
    \begin{aligned}
     \sum_{i=0}^{k-1} \left\lceil \frac{d+1}{q^i}\right\rceil &\geq\sum_{i=0}^{k-1} \left\lceil \frac{d_0+1}{q^i}\right\rceil\\
     &=\sum_{i=0}^{k-1} \left\lceil \frac{d_0}{q^i}\right\rceil+u_1\\
     &=n-(h-1)\frac{q^{u_0}-1}{q-1}+u_1\\
     &>n.
     \end{aligned}
\end{equation*}
    We have the optimal distance of $[n, k]_q$ is at most $d=d_0$.\\
    For $2)$, we have $\delta_q(k,d_0, 1)=g_q(k,d_0+1)-g_q(k,d_0)=i_h +u.$
    So
    \begin{equation*}
    \begin{aligned}
     \sum_{i=0}^{k-1} \left\lceil \frac{d+1}{q^i}\right\rceil &\geq\sum_{i=0}^{k-1} \left\lceil \frac{d_0+1}{q^i}\right\rceil\\
     &=\sum_{i=0}^{k-1} \left\lceil \frac{d_0}{q^i}\right\rceil+i_h+u\\
     &=n-(h-1)\frac{q^{u_0}-1}{q-1}+\sum_{i=1}^{k-u}\lfloor\frac{h}{q^i}\rfloor+i_h+u\\
     &>n.
     \end{aligned}
\end{equation*}
    We have the optimal distance of $[n, k]_q$ is at most $d=d_0$.

    \end{proof}

\begin{corollary}\label{corollary3.1}
    Let $h\leq q, u_1\leq \dots \leq u_h$. The linear code $\mathbf{C}$ is distance optimal if $u_1>(h-1)\frac{q^{u_0}-1}{q-1}$.
\end{corollary}

\begin{proof}

    If $h= q $ and $u_1=u_h$, we have
    \begin{equation}
        i_h+u_1=u_1+1>(h-1)\frac{q^{u_0}-1}{q-1}+1=(h-1)\frac{q^{u_0}-1}{q-1}+\sum_{i=1}^{k-u}\lfloor\frac{h}{q^i}\rfloor.
    \end{equation}
    In other cases, we have
    \begin{equation}
        u_1>(h-1)\frac{q^{u_0}-1}{q-1}.
    \end{equation}
   According to \textbf{Theorem} \ref{theorem3.3}, the conclusion is proved.
\end{proof}

\subsection{Weight distribution and subcode support weight distribution}

In this subsection, we determine weight distribution and subcode support weight distribution of codes constructed in the previous subsection.

\begin{lemma}\label{lemma3.3}[see \cite{Shi} Lemma 1,2,3.]
    Let $U_1$ be a $u_1$-dimensional subspace of $\mathbf{F}_q^k$ and $U_2$ be a $u_2$-dimensional subspace of $\mathbf{F}_q^k$ such that $dim(U_1 \cap U_2)=0.$
    \begin{enumerate}
        \item [(i)]
        Let $0\leq t \leq u_1 \leq k-1$ and $N_{l,t}^{k,u_1}$ denote the number of $l$-dimensional subspaces $V$ of $\mathbf{F}_q^k$ such that $dim(V\cap U_1)=t$. Then
        \begin{equation}
            N_{l,t}^{k,u_1}=q^{(u_1-t)(l-t)}Gaussian(k-u_1,l-t)_qGaussian(u_1,t)_q.
        \end{equation}
        \item [(ii)]
        Let $N_t(u_1,u_2,v_1,v_2)$ denote the number of $t+v_1+v_2$-dimensional subspaces $V$ of $U_1 \oplus U_2$ such that $dim(V \cap U_1)=v_1$ and $dim(V \cap U_2)=v_2$. Then
         \begin{equation*}
            N_t(u_1,u_2,v_1,v_2)=
            \begin{cases}
            \begin{aligned}
                (q^t-1)\cdots (q^t-q^{t-1})&Gaussian(u_1-v_1,t)_qGaussian(u_2-v_2,t)_q\\&Gaussian(u_1,v_1)_qGaussian(u_2,v_2)_q ,\ \mbox{if}\ t\neq 0.
                  \end{aligned}
                  \\
                Gaussian(u_1,v_1)_qGaussian(u_2,v_2)_q ,\  \mbox{if} \ t=0.
            \end{cases}
        \end{equation*}
        \item [(iii)]
         Let $N_{l,v_1,v_2}^{k,u_1,u_2}$ denote the number of $l$-dimensional subspaces $V$ of $\mathbf{F}_q^k$ such that $dim(V \cap U_1)=v_1$ and $dim(V \cap U_2)=v_2$. Then
         \begin{equation*}
         \begin{aligned}
         N_{l,v_1,v_2}^{k,u_1,u_2}=\sum_{t=0}^{t'}q^{(u_1+u_2-t-v_1-v_2)(l-t-v_1-v_2)}&N_t(u_1,u_2,v_1,v_2)\\&Gaussian(k-u_1-u_2,l-t-v_1-v_2)_q
         \end{aligned}
         \end{equation*}
         where $t'=min\{u_1-v_1,u_2-v_2\}$.
    \end{enumerate}
\end{lemma}

\begin{lemma}\label{lemma3.4}
    Let $U_1$ be a $u_1$-dimensional subspace of $\mathbf{F}_q^k$ and $U_2$ be a $u_2$-dimensional subspace of $\mathbf{F}_q^k$ such that $U_0=U_1 \cap U_2$ and $dim(U_0)=u_0<min\{u_1,u_2\}$.
    \begin{enumerate}
        \item [(i)]
        Let $N_t(u_0,u_1,u_2,v_0,v_1,v_2)$ denote the number of $t+v_1+v_2-v_0$-dimensional subspaces $V$ of $U_1 + U_2$ such that $dim(V \cap U_0)=v_0, dim(V \cap U_1)=v_1$ and $dim(V \cap U_2)=v_2$. Then
         \begin{equation*}
         \begin{aligned}
            N_t(u_0,u_1,u_2,v_0,v_1,v_2)=Gaussian(u_0,v_0)_q q^{(u_0-v_0)(t-v_1+v_2-v_0)}\\N_t(u_1-u_0,u_2-u_0,v_1-v_0,v_2-v_0).
         \end{aligned}
        \end{equation*}
        \item [(ii)]
         Let $N_{l,v_0,v_1,v_2}^{k,u_0,u_1,u_2}$ denote the number of $l$-dimensional subspaces $V$ of $\mathbf{F}_q^k$ such that $dim(V \cap U_0)=v_0, dim(V \cap U_1)=v_1$ and $dim(V \cap U_2)=v_2$. Then
         \begin{equation*}
         \begin{aligned}
         N_{l,v_0,v_1,v_2}^{k,u_0,u_1,u_2}=\sum_{t=0}^{t'}&q^{((u_1-v_1)+(u_2-v_2)-(u_0-v_0)-t)(l-t-v_1-v_2+v_0)}\\&Gaussian(k-u_1-u_2+u_0,l-t-u_1-u_2+v_0)_q N_t(u_0,u_1,u_2,v_0,v_1,v_2),
         \end{aligned}
         \end{equation*}
         where $t'=min\{u_1-v_1,u_2-v_2\}$.
    \end{enumerate}
\end{lemma}

\begin{proof}
\begin{enumerate}
    \item [(i)]
    Assume that
    \begin{equation*}
        S=\{ V\leq U_1\cap U_2|dim(V\cap U_i)=v_i,i\in\{0,1,2\},dim(V\cap U_1+U_2)=t+v_1+v_2-v_0\}.
    \end{equation*}
    Then we have $S=\bigcup_{V_0}S_{V_0}$ where $V_0$ is any $v_0$-dimensional subspace of $U_0$ and
    \begin{equation*}
        S_{V_0}=\{ V|dim(V\cap U_i)=v_i,i\in\{1,2\},dim(V\cap U_1+U_2)=t+v_1+v_2-v_0,V\cap U_0=V_0\}.
    \end{equation*}
    We consider the quotient spaces.\\
    Let $\overline{V}=(V+U_0)/U_0$.
    It is easy to see that $\overline{U_1}\cap \overline{U_2}=0$.
    Assume that
    \begin{equation*}
        S_{V_0}'=\{V/V_0|V\in S_{V_0}\}, S_{V_0}''=\{\overline{V}|V\in S_{V_0}\}
    \end{equation*}
     and
     \begin{equation*}
         \overline{S}=\{\overline{V}\leq \overline{\mathbf{F}_q^k}|\overline{V}\leq \overline{U_1}+\overline{U_2},dim(\overline{V}\cap \overline{U_1})=v_1-v_0,dim(\overline{V}\cap \overline{U_2})=v_2-v_0\}.
     \end{equation*}
     It is obvious that we can define two maps $\varphi:S_{V_0} \to S_{V_0}'',V \mapsto V/V_0$ and $\psi:S_{V_0}' \to S_{V_0}'',V/V_0 \mapsto V/U_0$ since $\varphi$ is bijective.\\
    If $V\in S_{V_0}$, we have $dim(\overline{V}\cap \overline{U_1})=v_1-v_0, dim(\overline{V}\cap \overline{U_2})=v_2-v_0$ and $dim\overline{V}=t+v_1+v_2-2v_0$ which means that $S_{V_0}''\subseteq\overline{S}$.\\
    For any $\overline{V}\in \overline{S}$, let $\overline{\boldsymbol{\alpha}_1},\dots,\overline{\boldsymbol{\alpha}_{v_1-v_0}},\overline{\boldsymbol{\beta}_1},\dots,\overline{\boldsymbol{\beta}_{v_2-v_0}},\overline{\boldsymbol{\gamma}_1},\dots,\overline{\boldsymbol{\gamma}_{t}}$ be a base of $\overline{V}$ \\where $\overline{\boldsymbol{\alpha}_1},\dots,\overline{\boldsymbol{\alpha}_{v_1-v_0}}$ is a base of $ \overline{V}\cap \overline{U_1}$, $\overline{\boldsymbol{\beta}_1},\dots,\overline{\boldsymbol{\beta}_{v_2-v_0}}$ is a base of $ \overline{V}\cap \overline{U_2}$. It is not difficult to verify that
    \begin{equation*}
    \begin{aligned}
        \psi^{-1}(\overline{V})=\{V/V_0 |& V/V_0=\\
        &<(\boldsymbol{\alpha}_1+V_0)+(\boldsymbol{\delta}_1+V_0),\cdots,(\boldsymbol{\alpha}_{v_1-v_0}+V_0)+(\boldsymbol{\delta}_{v_1-v_0}+V_0),\\&(\boldsymbol{\beta}_1+V_0)+(\boldsymbol{\delta}_{v_1-v_0+1}+V_0),\cdots,(\boldsymbol{\beta}_{v_2-v_0}+V_0)+(\boldsymbol{\delta}_{v_1+v_2-2v_0}+V_0),\\&(\boldsymbol{\gamma}_1+V_0)+(\boldsymbol{\delta}_{v_1+v_2-2v_0+1}+V_0),\cdots,(\boldsymbol{\gamma}_t+V_0)+(\boldsymbol{\delta}_{t+v_1+v_2-2v_0}+V_0)>,\\&\boldsymbol{\delta}_i+V_0\in U_0/V_0\}.
    \end{aligned}
    \end{equation*}
    It is obvious to verify that
    \begin{equation*}
    \begin{aligned}
        <&(\boldsymbol{\alpha}_1+V_0)+(\boldsymbol{\delta}_1+V_0),\cdots,(\boldsymbol{\alpha}_{v_1-v_0}+V_0)+(\boldsymbol{\delta}_{v_1-v_0}+V_0),\\&(\boldsymbol{\beta}_1+V_0)+(\boldsymbol{\delta}_{v_1-v_0+1}+V_0),\cdots,(\boldsymbol{\beta}_{v_2-v_0}+V_0)+(\boldsymbol{\delta}_{v_1+v_2-2v_0}+V_0),\\&(\boldsymbol{\gamma}_1+V_0)+(\boldsymbol{\delta}_{v_1+v_2-2v_0+1}+V_0),\cdots,(\boldsymbol{\gamma}_t+V_0)+(\boldsymbol{\delta}_{t+v_1+v_2-2v_0}+V_0)>\\
        =<&(\boldsymbol{\alpha}_1+V_0)+(\boldsymbol{\delta}_1'+V_0),\cdots,(\boldsymbol{\alpha}_{v_1-v_0}+V_0)+(\boldsymbol{\delta}_{v_1-v_0}'+V_0),\\&(\boldsymbol{\beta}_1+V_0)+(\boldsymbol{\delta}_{v_1-v_0+1}'+V_0),\cdots,(\boldsymbol{\beta}_{v_2-v_0}+V_0)+(\boldsymbol{\delta}_{v_1+v_2-2v_0}'+V_0),\\&(\boldsymbol{\gamma}_1+V_0)+(\boldsymbol{\delta}_{v_1+v_2-2v_0+1}'+V_0),\cdots,(\boldsymbol{\gamma}_t+V_0)+(\boldsymbol{\delta}_{t+v_1+v_2-2v_0}'+V_0)>
    \end{aligned}
    \end{equation*}
    if and only if $\boldsymbol{\delta}_i+V_0=\boldsymbol{\delta}_i'+V_0$ for $i=1,\cdots, v_1+v_2-2v_0+t$.
    Then we have for any $\overline{V}\in \overline{S}$,
    \begin{equation*}
        |\psi^{-1}(\overline{V})|=q^{(u_0-v_0)(t+v_1+v_2-2v_0)}.
    \end{equation*}
    Hence we have
    \begin{equation*}
        \begin{aligned}
            N_t(u_0,u_1,u_2,v_0,v_1,v_2)&=|S|
            =\sum_{V_0}|S_{V_0}|
            =\sum_{V_0}|S_{V_0}'|\\
            &=Gaussian(u_0,v_0)_q |\psi^{-1}(\overline{V})||\overline{S}|\\
            &=Gaussian(u_0,v_0)_q q^{(u_0-v_0)(t-v_1+v_2-v_0)}\\ &\ \ \ \ \ N_t(u_1-u_0,u_2-u_0,v_1-v_0,v_2-v_0).
        \end{aligned}
    \end{equation*}
    \item[(ii)]
    It is similar to the proof of (i) and \textbf{Lemma} \ref{lemma3.3} (iii).
    \end{enumerate}
\end{proof}

\begin{theorem}\label{theorem3.4}
    Assume $u_0<u_1\leq \dots \leq u_h, h\leq q, k\geq \sum_{i=1}^h u_i-(h-1)u_0, q^k - q^{k-1} > q^{u_0}-1+\sum_{i=1}^h (q^{u_i}-q^{u_0})$  and $u_1>(h-1)\frac{q^{u_0}-1}{q-1}$, we construct an optimal linear $[n,k,d]_q$ code $\mathbf{C}_3$ with
    \begin{equation*}
        n=\frac{q^k-q^{u_0}-\sum_{i=1}^h (q^{u_i}-q^{u_0})}{q-1}
    \end{equation*}
    and
    \begin{equation*}
        d=q^{k-1}-\sum_{i=1}^h q^{u_i-1}.
    \end{equation*}

\begin{enumerate}
\item[(a)]
    The weight distribution of the linear code $\mathbf{C}_3$ is determined in the following table.
\begin{center}
\begin{tabular}{|c|c|}
\hline
weight & multiplicity \\
\hline
0 & 1 \\
\hline
$q^{k-1}+(h-1) q^{u_0-1}-\sum_{i=1}^h q^{u_i-1}$ &$(q^{u_0}-1)q^{k-u_0}$ \\
\hline
$q^{k-1}-\sum_{i=1}^h q^{u_i-1}$ & $(\prod_{i=1}^h(q^{u_i-u_0}-1))q^{k+(h-1)u_0-\sum_{i=1}^h u_i}$\\
\hline
$q^{k-1}-\sum_{i=2}^h q^{u_i-1}$ & $(\prod_{i=2}^h(q^{u_i-u_0}-1))q^{k+(h-1)u_0-\sum_{i=1}^h u_i}$\\
\hline
\vdots & \vdots \\
\hline
$q^{k-1}-q^{u_h-1}$ & $(q^{u_h-u_0}-1)q^{k+(h-1)u_0-\sum_{i=1}^h u_i}$\\
\hline
$q^{k-1}$ & $q^{k+(h-1)u_0-\sum_{i=1}^h u_i}-1$\\
\hline
\end{tabular}
\end{center}
\item [(b)]
Let $h=2$. The $r$-GHW of $\mathbf{C}_3$ is determined as follows.
\begin{equation}
  d_r(\mathbf{C}_3)= \begin{cases}
      \frac{q^k-q^{k-r}-(q^{u_1}-q^{u_1-r})-(q^{u_2}-q^{u_2-r})}{q-1},1\leq r \leq u_1-u_0\\
      \frac{q^k-q^{k-r}-q^{u_1}-(q^{u_2}-q^{u_2-r})+q^{u_0}}{q-1},u_1-u_0\leq r \leq u_2\\
      \frac{q^k-q^{k-r}-q^{u_1}-q^{u_2}+q^{u_0}+1}{q-1},u_2\leq r <k\\
  \end{cases}
\end{equation}

\item [(c)]
Let $h=2$. The $r$-SSWD of $\mathbf{C}_3$ is
\begin{equation}
    A_j^r(\mathbf{C}_3)=\sum_{(v_0,v_1,v_2)\in Set} N_{k-r,v_0,v_1,v_2}^{k,u_0,u_1,u_2},
\end{equation}
where
\begin{equation*}
 \begin{aligned}
     Set=\{ (v_0,v_1,v_2)\in \mathbf{Z}^3|&\frac{q^{v_1}-1}{q-1}+\frac{q^{v_2}-1}{q-1}-\frac{q^{v_0}-1}{q-1}=j-\frac{q^k-q^{u_1}-q^{u_2}+q^{u_0}-q^{k-r}+1}{q-1};\\
         &max\{u_i-r,0\}\leq v_i\leq max\{k-r,u_i\},   i=0,1,2;\\
         &v_0\leq v_1, v_0\leq v_2;\\
         &v_1-v_0\leq u_1-u_0,v_2-v_0\leq u_2-u_0;\\
         &v_1+v_2-v_0 \leq k-r
                                   \}.
 \end{aligned}
\end{equation*}
\end{enumerate}
\end{theorem}
\begin{proof}
    Let $\mathbf{e}_i$ be the vector with all $0$ coordinates except a $1$ coordinate at the $i$th position. Assume that $U_i$ is the linear space generated by
    $\{\mathbf{e}_1,\dots,\mathbf{e}_{u_0},\mathbf{e}_{u_0+\sum_{j=0}^{i-1}(u_j-u_0)+1},\dots,\mathbf{e}_{\sum_{j=0}^{i}u_j-iu_0}\}$  when $i=1,\dots, h$. Let $U$ be the set of all vectors among these $q^k-1$ nonzero vectors in  $\mathbf{F}_q^k$ by deleting these nonzero vectors in $\bigcup_{i=1}^h U_i$.
    If vector $\mathbf{v}$ is a column vector of $U$, $\lambda\mathbf{v}$ is also a column vector of $U$ where $\lambda \in \mathbf{F}_q^* $. If we pick up only one vector from $\{\lambda\mathbf{v}|\lambda \in \mathbf{F}_q^* \}$ for every vector $\mathbf{v}\in U$, these vectors form a $k\times \frac{|U|}{q-1}$ matrix $\widetilde{G}$.
    According to \textbf{Corollary} \ref{corollary3.1}, $\mathbf{C}_3$ is an optimal linear code with length
    \begin{equation*}
        n=\frac{q^k-q^{u_0}-\sum_{i=1}^h (q^{u_i}-q^{u_0})}{q-1}
    \end{equation*}
    and minimum weight
    \begin{equation*}
        d=q^{k-1}-\sum_{i=1}^h q^{u_i-1}.
    \end{equation*}
\begin{enumerate}
\item[(a)]
    For any nonzero codeword $\mathbf{c}\in \mathbf{C}_3$, there exists unique nonzero vector $\mathbf{a}=(a_1,\dots,a_k)\in \mathbf{F}_q^k$, such that $\mathbf{c}=\mathbf{a}\cdot \widetilde{G}$ and $H=\langle \mathbf{a}\rangle^\perp$.
    Let $Set_H=\{i|i\in\{1,\dots,h\},dim(H\cap U_i)=u_i\}$.
    The several cases are listed in the following table.
\begin{center}
\begin{tabular}{|c|c|c|}
\hline
condition of $H$ and $Set_H$&weight & multiplicity \\
\hline
&0 & 1 \\
\hline
$dim(H\cap U_0)=u_0-1$&\makecell{$q^{k-1}+(h-1) q^{u_0-1}$\\ $-\sum_{i=1}^h q^{u_i-1}$ }&$(q^{u_0}-1)q^{k-u_0}$ \\
\hline
\makecell{$dim(H\cap U_0)=u_0$\\$Set_H=\emptyset$}&$q^{k-1}-\sum_{i=1}^h q^{u_i-1}$ & $(\prod_{i=1}^h(q^{u_i-u_0}-1))q^{k+(h-1)u_0-\sum_{i=1}^h u_i}$\\
\hline
\makecell{$dim(H\cap U_0)=u_0$\\$Set_H=\{1\}$}&$q^{k-1}-\sum_{i=2}^h q^{u_i-1}$ & $(\prod_{i=2}^h(q^{u_i-u_0}-1))q^{k+(h-1)u_0-\sum_{i=1}^h u_i}$\\
\hline
\vdots &\vdots & \vdots \\
\hline
\makecell{$dim(H\cap U_0)=u_0$\\$Set_H=\{1,\dots ,h-1\}$}&$q^{k-1}-q^{u_h-1}$ & $(q^{u_h-u_0}-1)q^{k+(h-1)u_0-\sum_{i=1}^h u_i}$\\
\hline
\makecell{$dim(H\cap U_0)=u_0$\\$Set_H=\{1,\dots ,h\}$}&$q^{k-1}$ & $q^{k+(h-1)u_0-\sum_{i=1}^h u_i}-1$\\
\hline
\end{tabular}
\end{center}
\begin{enumerate}[(i)]
    \item If $dim(U_0\cap H)=u_0-1$, then $dimU_i\cap H=u_i-1$ for $i=1,\dots,h$.
    We have
    \begin{equation*}
    \begin{aligned}
      wt(\mathbf{c})&= \frac{|U|-|H-H\cap (U_1\cup\dots\cup U_h)|}{q-1}\\
                   &= \frac{|U|-|H|+|H\cap U_1|+\dots+|H\cap U_h|-(h-1)|H\cap U_0|}{q-1}\\
                   &= q^{k-1}+(h-1) q^{u_0-1}-\sum_{i=1}^h q^{u_i-1}.
    \end{aligned}
    \end{equation*}
    And we know that
     \begin{equation*}
    \begin{aligned}
      &|\{\mathbf{a}\in \mathbf{F}_q^k|\mathbf{a}\neq \mathbf{0} \ , \ U_0\nsubseteq\langle \mathbf{a}\rangle^\perp\}|\\=&q^k-1-|\{\mathbf{a}\in \mathbf{F}_q^k|\mathbf{a}\neq \mathbf{0} \ , \ U_0\subseteq\langle \mathbf{a}\rangle^\perp\}|\\=&q^k-1-|\{\mathbf{a}\in \mathbf{F}_q^k|\mathbf{a}\neq \mathbf{0} \ , a_1=\dots=a_{u_0} =0| \\=&(q^{u_0}-1)q^{k-u_0}.
    \end{aligned}
    \end{equation*}
    \item If $dim(U_0\cap H)=u_0$, let $Set'_H=\{i|i\in\{1,\dots,h\},dim(H\cap U_i)=u_i-1\}$. We know that $dim(U_i\cap H)=u_i$ when $i\in Set_H$ and $dim(U_i\cap H)=u_i-1$ when $i\in Set'_H$. Then we have
    \begin{equation*}
    \begin{aligned}
      wt(\mathbf{c})&= \frac{|U|-|H-H\cap (U_1\cup\dots\cup U_h)|}{q-1}\\
                   &= \frac{|U|-|H|+|H\cap U_1|+\dots+|H\cap U_h|-(h-1)|H\cap U_0|}{q-1}\\
                   &= q^{k-1}-\sum_{i\in Set'_H} q^{u_i-1} .
    \end{aligned}
    \end{equation*}
    And we know that
     \begin{equation*}
     \begin{aligned}
      &|\{\mathbf{a}\in \mathbf{F}_q^k|\mathbf{a}\neq \mathbf{0} , \ U_0\subseteq\langle \mathbf{a}\rangle^\perp,\ U_i\subseteq \langle \mathbf{a}\rangle^\perp,\ U_j\nsubseteq \langle \mathbf{a}\rangle^\perp \
      ,\ i\in Set_H\ ,\ j\in Set'_H\}|\\=&(\prod_{i\in Set'_H}(q^{u_i-u_0}-1))q^{k+(h-1)u_0-\sum_{i=1}^h u_i}.
    \end{aligned}
    \end{equation*}
\end{enumerate}
\item [(b)]
If $1\leq r\leq u_1-u_0$, we have
\begin{equation*}
    \begin{aligned}
        d_r(\mathbf{C}_3)&=n-max\{m_{\widetilde{G}}(V)|V\in SUB^{k-r}(\mathbf{F}_q^k)\}\\
        &=n-\frac{max\{m_{G}(V)|V\in SUB^{k-r}(\mathbf{F}_q^k)\}}{q-1}\\
        &=n-\frac{max\{|V|-|V\cap U_1|-|V\cap U_2|+|V\cap U_0||V\in SUB^{k-r}(\mathbf{F}_q^k)\}}{q-1}\\
        &\geq \frac{q^k-q^{k-r}-(q^{u_1}-q^{u_1-r})-(q^{u_2}-q^{u_2-r})}{q-1}.
\end{aligned}
\end{equation*}
Let          $\boldsymbol{\alpha}_1=\mathbf{e}_1,\dots,\boldsymbol{\alpha}_{u_0}=\mathbf{e}_{u_0},\boldsymbol{\alpha}_{u_0+1}=\mathbf{e}_{u_0+1},\dots,\boldsymbol{\alpha}_{u_1-r}=\mathbf{e}_{u_1-r},\boldsymbol{\alpha}_{u_1-r+1}=\mathbf{e}_{u_1+1},\dots,\\ \boldsymbol{\alpha}_{u_1+u_2-u_0-2r}=\mathbf{e}_{u_1+u_2-u_0-r},\boldsymbol{\alpha}_{u_1+u_2-u_0-2r+1}=\mathbf{e}_{u_1-r+1}+\mathbf{e}_{u_1+u_2-u_0-r+1},\dots , \\
\boldsymbol{\alpha}_{u_1+u_2-u_0-r}=\mathbf{e}_{u_1}+\mathbf{e}_{u_1+u_2-u_0},\boldsymbol{\alpha}_{u_1+u_2-u_0-r+1}=\mathbf{e}_{u_1+u_2-u_0+1},\dots , \boldsymbol{\alpha}_{k-r}=\mathbf{e}_k$ and $V$ be the subspace generated by $\boldsymbol{\alpha}_1,\dots,\boldsymbol{\alpha}_{k-r}$. We have
\begin{equation}
    d_r(\mathbf{C}_3)=\frac{q^k-q^{k-r}-(q^{u_1}-q^{u_1-r})-(q^{u_2}-q^{u_2-r})}{q-1},\ 1\leq r\leq u_1-u_0.
\end{equation}
The rest cases are similar.
\item[(c)]
It is easy to verify that
\begin{equation*}
    \begin{aligned}
        A_j^r(\mathbf{C}_3)&=|\{U\in SUB^r(\mathbf{C}_3)||supp(U)|=j        \}|\\
        &=|\{V\in SUB^{k-r}(\mathbf{F}_q^k)|n-m_{\widetilde{G}}(V)=j        \}|\\
        &=|\{V\in SUB^{k-r}(\mathbf{F}_q^k)|\frac{q^{v_1}-1}{q-1}+\frac{q^{v_2}-1}{q-1}-\frac{q^{v_0}-1}{q-1}\\
        &\ \ \ \ \ \ \ \  \ \  \ \ \ \ \  \ \ \ \ \ \ \ \  \ \ \ \ \ \ \ \ \ \ \ \ \ \ \ \ =j-\frac{q^k-q^{u_1}-q^{u_2}+q^{u_0}-q^{k-r}+1}{q-1}        \}\\
        &=\sum_{(v_0,v_1,v_2)\in Set} N_{k-r,v_0,v_1,v_2}^{k,u_0,u_1,u_2},
    \end{aligned}
\end{equation*}
where
\begin{equation*}
 \begin{aligned}
     Set=\{ (v_0,v_1,v_2)\in \mathbf{Z}^3|&\frac{q^{v_1}-1}{q-1}+\frac{q^{v_2}-1}{q-1}-\frac{q^{v_0}-1}{q-1}=j-\frac{q^k-q^{u_1}-q^{u_2}+q^{u_0}-q^{k-r}+1}{q-1};\\
         &max\{u_i-r,0\}\leq v_i\leq max\{k-r,u_i\}, \mbox{for}\ i=0,1,2;\\
         &v_0\leq v_1, v_0\leq v_2;\\
         &v_1-v_0\leq u_1-u_0,v_2-v_0\leq u_2-u_0;\\
         &v_1+v_2-v_0 \leq k-r
                                   \}.
 \end{aligned}
\end{equation*}

\end{enumerate}
The conclusion is proved.
\end{proof}

\begin{example}
    With the construction in \textbf{Theorem} \ref{theorem3.4}, let $q=2,k=6,h=2$ and $u_0=2,u_1=4,u_2=4$. We construct an optimal linear $[36,6,16]_2$ code with the weight distribution as the following table.
\begin{center}
\begin{tabular}{|c|c|}
\hline
weight & multiplicity \\
\hline
0 & 1 \\
\hline
16 & 9\\
\hline
18 & 48 \\
\hline
24 & 6\\
\hline

\end{tabular}
\end{center}

\end{example}

\section{Locally recoverable codes and their CM defects}
In this section, $\mathbf{C}_1$ is the code constructed in Theorem \ref{theorem2.4}, $\mathbf{C}_2$ is the code constructed in Theorem \ref{theorem2.5} and $\mathbf{C}_3$ is code constructed in Theorem \ref{theorem3.4}.

\begin{lemma}\label{lemma4.1}
    We assume that $q \geq 3$. There exist two nonzero vectors $\boldsymbol{\beta} ,\boldsymbol{\gamma}\in \mathbf{F}_q^{k_q}$ such that $\boldsymbol{\beta} +\boldsymbol{\gamma}=\boldsymbol{\alpha}$ for any $\boldsymbol{\alpha}\in \mathbf{F}_q^{k_q}$ where $k_2\geq 2$ and $k_q\geq 1$. In particular, there exist two nonzero vectors $\boldsymbol{\beta} ,\boldsymbol{\gamma}\in \mathbf{F}_q^{k_q}$ such that $\boldsymbol{\beta} +\boldsymbol{\gamma}=\boldsymbol{\alpha}$ and $\boldsymbol{\alpha},\boldsymbol{\beta},\boldsymbol{\gamma}$ are pairwise linearly independent for any $\mathbf{0}\neq \boldsymbol{\alpha}\in \mathbf{F}_q^{k_q}$ where $k_q\geq 2$.
\end{lemma}
\begin{proof}
    There are $q^{k_q}-1\geq 2$ nonzero vectors in $\mathbf{F}_q^{k_q}$. Let $\mathbf{0},\boldsymbol{\alpha}\neq \boldsymbol{\beta} \in\mathbf{F}_q^{k_q} $ and $\boldsymbol{\gamma}=\boldsymbol{\alpha} -\boldsymbol{\beta}$. We have $\boldsymbol{\beta},\boldsymbol{\gamma} \neq 0$ and $\boldsymbol{\beta} +\boldsymbol{\gamma}=\boldsymbol{\alpha}$.\\
    In particular if $k_q\geq2$, we have $q^{k_q}-q>0$ vectors linearly independently with $\boldsymbol{\alpha}$.
\end{proof}

\begin{theorem}\label{theorem4.1}
    The code $\mathbf{C}_1,\mathbf{C}_2,\mathbf{C}_3$ are LRC codes with the locality $r=2$.
\end{theorem}
\begin{proof}
 Let $\mathbf{e}_i$ be the vector with all $0$ coordinates except a $1$ coordinate at the $i$th position.
    \begin{enumerate}
        \item [(i)] We consider the code $\mathbf{C}_1$.\\
        Let $\boldsymbol{\alpha}=\begin{pmatrix}\boldsymbol{\alpha}_1\\\vdots\\\boldsymbol{\alpha}_h\\\boldsymbol{\alpha}_{h+1}\end{pmatrix}$ be a column vector of $\widetilde{G}$ where  $\boldsymbol{\alpha}_i\in\mathbf{F}_q^{u} $, $i=1,\dots,h$.
We only need to find two column vectors $\boldsymbol{\beta}$ and $\boldsymbol{\gamma}$ in $G$ such that $\boldsymbol{\alpha},\boldsymbol{\beta},\boldsymbol{\gamma}$ are pairwise linearly independent and $\boldsymbol{\alpha}=\boldsymbol{\beta}+\boldsymbol{\gamma}$ in the following two cases.
\begin{enumerate}
    \item [(a)]$k>uh$\\
    If $\boldsymbol{\alpha}_{h+1}=\mathbf{0}$, let $\boldsymbol{\beta}=\mathbf{e}_k,\boldsymbol{\gamma}=\boldsymbol{\alpha}-\boldsymbol{\beta}$.\\
    If $\boldsymbol{\alpha}_{h+1}\neq\mathbf{0}$, let $\boldsymbol{\beta}=\mathbf{e}_1+\dots+\mathbf{e}_{uh},\boldsymbol{\gamma}=\boldsymbol{\alpha}-\boldsymbol{\beta}$.
    \item [(b)]$k=uh$\\
    There exist $i,j\in \{1,\dots,h\}$ such that $i<j$ and $\boldsymbol{\alpha}_i,\boldsymbol{\alpha}_j\neq \mathbf{0}$. Let $\boldsymbol{\beta}_i,\boldsymbol{\beta}_j,\boldsymbol{\gamma}_i,\boldsymbol{\gamma}_j$ be the nonzero vectors such that $\boldsymbol{\alpha}_i=\boldsymbol{\beta}_i+\boldsymbol{\gamma}_i,\boldsymbol{\alpha}_j=\boldsymbol{\beta}_j+\boldsymbol{\gamma}_j$, where $\boldsymbol{\alpha}_i,\boldsymbol{\beta}_i,\boldsymbol{\gamma}_i$ are pairwise linearly independent and $\boldsymbol{\alpha}_j,\boldsymbol{\beta}_j,\boldsymbol{\gamma}_j$ are pairwise linearly independent.\\
   Let $\boldsymbol{\beta}=\begin{pmatrix}\mathbf{0}\\\vdots\\\boldsymbol{\beta}_i\\\vdots\\\boldsymbol{\beta}_j\\\vdots\\\mathbf{0}\end{pmatrix},\boldsymbol{\gamma}=\boldsymbol{\alpha}-\boldsymbol{\beta}$.\\
\end{enumerate}
Hence, we have the locality of the linear code $\mathbf{C}_1$ is 2.
    \item [(ii)] We consider the code $\mathbf{C}_2$.\\
    We prove that every column vector $\boldsymbol{\alpha}=(a_1 ,\dots,a_k )^T$ in $\widetilde{G}$ can be linearly express by other two column vectors in $\widetilde{G}$.
We only need to find two column vectors $\boldsymbol{\beta}$ and $\boldsymbol{\gamma}$ in $G$ such that $\boldsymbol{\alpha},\boldsymbol{\beta},\boldsymbol{\gamma}$ are pairwise linearly independent and $\boldsymbol{\alpha}=\boldsymbol{\beta}+\boldsymbol{\gamma}$ in the following two cases.
    \begin{enumerate}
    \item[(a)] $k>8$\\
    If there exists $9\leq i\leq k$ such that $a_i=0$, let $\boldsymbol{\beta}=\mathbf{e}_i,\boldsymbol{\gamma}=\boldsymbol{\alpha}-\boldsymbol{\beta}$.\\
    If $a_i \neq0$ for $9\leq i\leq k$, let $\boldsymbol{\beta}=\mathbf{e}_2 , \boldsymbol{\gamma}=\boldsymbol{\alpha}-\boldsymbol{\beta}$.
    \item[(b)] $k=8$\\
    If there exists $ i\in \{2,4,6,8\}$ such that $a_i=0$, let $\boldsymbol{\beta}=\mathbf{e}_i,\boldsymbol{\gamma}=\boldsymbol{\alpha}-\boldsymbol{\beta}$.\\
    If not, let $\boldsymbol{\beta}=a_2\mathbf{e}_2,\boldsymbol{\gamma}=\boldsymbol{\alpha}-\boldsymbol{\beta}$.
    \end{enumerate}
Hence, we have the locality of the linear code $\mathbf{C}_2$ is 2.
\item [(iii)] We consider the code $\mathbf{C}_3$.\\
It is obvious that if $q=2$ and $k=u_1+u_2-u_0$, we have $u_1-u_0\geq2.$ Otherwise, we have $2^{u_2}=q^k-q^{k-1}\leq q^{u_0}-1+\sum_{i=1}^h (q^{u_i}-q^{u_0})=2^{u_2}+2^{u_0}-1.$ This leads to the contradiction.\\
Let $\boldsymbol{\alpha}=\begin{pmatrix}\boldsymbol{\alpha}_0\\\vdots\\\boldsymbol{\alpha}_h\\\boldsymbol{\alpha}_{h+1}\end{pmatrix}$ be a column vector of $\widetilde{G}$ where $\boldsymbol{\alpha}_0\in \mathbf{F}_q^{u_0} $ and $\boldsymbol{\alpha}_i\in\mathbf{F}_q^{u_i-u_0} $ for $i\in\{1,\dots,h\}.$
We only need to find two column vectors $\boldsymbol{\beta}$ and $\boldsymbol{\gamma}$ in $G$ such that $\boldsymbol{\alpha},\boldsymbol{\beta},\boldsymbol{\gamma}$ are pairwise linearly independent and $\boldsymbol{\alpha}=\boldsymbol{\beta}+\boldsymbol{\gamma}$ in the following two cases.
\begin{enumerate}
    \item [(a)]$k>\sum_{i=1}^h u_i-(h-1)u_0$\\
    If $\boldsymbol{\alpha}_{h+1}=\mathbf{0}$, let $\boldsymbol{\beta}=\mathbf{e}_k,\boldsymbol{\gamma}=\boldsymbol{\alpha}-\boldsymbol{\beta}$.\\
    If $\boldsymbol{\alpha}_{h+1}\neq\mathbf{0}$, let $\boldsymbol{\beta}=\mathbf{e}_1+\dots+\mathbf{e}_{\sum_{i=1}^h u_i-(h-1)u_0},\boldsymbol{\gamma}=\boldsymbol{\alpha}-\boldsymbol{\beta}$.
    \item [(b)]$k=\sum_{i=1}^h u_i-(h-1)u_0$\\
    We have there exist $i,j\in \{1,\dots,h\}$ such that $i<j$ and $\boldsymbol{\alpha}_i,\boldsymbol{\alpha}_j\neq \mathbf{0}$. Let $\boldsymbol{\beta}_i,\boldsymbol{\beta}_j,\boldsymbol{\gamma}_i,\boldsymbol{\gamma}_j$ be the nonzero vectors such that $\boldsymbol{\alpha}_i=\boldsymbol{\beta}_i+\boldsymbol{\gamma}_i,\boldsymbol{\alpha}_j=\boldsymbol{\beta}_j+\boldsymbol{\gamma}_j$.\\
    If $\boldsymbol{\alpha}_0=\mathbf{0}$, let $\boldsymbol{\beta}=\mathbf{e}_1+\begin{pmatrix}\mathbf{0}\\\vdots\\\boldsymbol{\beta}_i\\\vdots\\\boldsymbol{\beta}_j\\\vdots\\\mathbf{0}\end{pmatrix},\boldsymbol{\gamma}=\boldsymbol{\alpha}-\boldsymbol{\beta}$.\\
    If $\boldsymbol{\alpha}_0\neq\mathbf{0}$, let $\boldsymbol{\beta}=\begin{pmatrix}\mathbf{0}\\\vdots\\\boldsymbol{\beta}_i\\\vdots\\\boldsymbol{\beta}_j\\\vdots\\\mathbf{0}\end{pmatrix},\boldsymbol{\gamma}=\boldsymbol{\alpha}-\boldsymbol{\beta}$.\\
\end{enumerate}
Hence, we have the locality of the linear code $\mathbf{C}_3$ is 2.
    \end{enumerate}

\end{proof}

\begin{theorem}\label{theorem4.2}
    The CM defect of the linear code $\mathbf{C}_2$ can be upper bounded as follows.\\
    If $q=2$, the CM defect of the linear code $\mathbf{C}_2$ is upper bounded by $2$.\\
    If $q\geq 3$, the CM defect of the linear code $\mathbf{C}_2$ is upper bounded by $1$.
\end{theorem}

\begin{proof}
    If $q=2$, we have
\begin{equation*}
    \begin{aligned}
        min_{t>0}\{  rt+k_{op}^{(q)} (n-(r+1)t,d) \}&=min_{t>0}\{  2t+k_{op}^{(q)} (n-3t,d) \}\\
        &\leq 2t+k_{op}^{(q)} (n-3t,d)|_{t=1} \\
        &=2+k_{op}^{(q)} (n-3,d)\\
        &\leq 2+k,
    \end{aligned}
\end{equation*}
since $ \sum_{i=0}^{k-1} \left\lceil \frac{d}{q^i}\right\rceil \leq n-3 < \sum_{i=0}^{k} \left\lceil \frac{d}{q^i}\right\rceil.$\\
If $q\geq 3$, we have
\begin{equation*}
    \begin{aligned}
        min_{t>0}\{  rt+k_{op}^{(q)} (n-(r+1)t,d) \}&=min_{t>0}\{  2t+k_{op}^{(q)} (n-3t,d) \}\\
        &\leq 2t+k_{op}^{(q)} (n-3t,d)|_{t=1} \\
        &=2+k_{op}^{(q)} (n-3,d)\\
        &\leq 2+k-1=k+1,
    \end{aligned}
\end{equation*}
since $ \sum_{i=0}^{k-2} \left\lceil \frac{d}{q^i}\right\rceil \leq n-3 < \sum_{i=0}^{k-1} \left\lceil \frac{d}{q^i}\right\rceil.$
\end{proof}

\begin{theorem}\label{theorem4.3}
\begin{enumerate}
    \item [(a)]
    Let $u_0=1,h=2.$ If $q\geq3$ or $q=2,u_1<u_2$, the linear code $\mathbf{C}_3$ meets the CM bound.
    \item [(b)]
    Let $u_0=1,h=3.$ If $q\geq4$, the CM defect of the linear code $\mathbf{C}_3$ is upper bounded by 1.
\end{enumerate}
\end{theorem}

\begin{proof}
\begin{enumerate}
    \item [(a)]
We have
\begin{equation*}
    \begin{aligned}
        min_{t>0}\{  rt+k_{op}^{(q)} (n-(r+1)t,d) \}&=min_{t>0}\{  2t+k_{op}^{(q)} (n-3t,d) \}\\
        &\leq 2t+k_{op}^{(q)} (n-3t,d)|_{t=1} \\
        &=2+k_{op}^{(q)} (n-3,d)\\
        &\leq 2+k-2=k,
    \end{aligned}
\end{equation*}
since $ \sum_{i=0}^{k-3} \left\lceil \frac{d}{q^i}\right\rceil \leq n-3 < \sum_{i=0}^{k-2} \left\lceil \frac{d}{q^i}\right\rceil.$
\item [(b)]
We have
\begin{equation*}
    \begin{aligned}
        min_{t>0}\{  rt+k_{op}^{(q)} (n-(r+1)t,d) \}&=min_{t>0}\{  2t+k_{op}^{(q)} (n-3t,d) \}\\
        &\leq 2t+k_{op}^{(q)} (n-3t,d)|_{t=1} \\
        &=2+k_{op}^{(q)} (n-3,d)\\
        &\leq 2+k-1=k+1,
    \end{aligned}
\end{equation*}
since $ \sum_{i=0}^{k-2} \left\lceil \frac{d}{q^i}\right\rceil \leq n-3 < \sum_{i=0}^{k-1} \left\lceil \frac{d}{q^i}\right\rceil.$
\end{enumerate}
\end{proof}

\section{Conclusion}

In this paper, infinite families of optimal codes with positive Griesmer defects were constructed. The method is inspired from the construction of \cite{Solomon} by Solomon-Stiffler. However these optimal codes have positive Griesmer defects and are not equivalent to Solomon-Stiffler codes. Weight distributions and subcode support weight distributions were determined. On the other hand, these codes are LRC codes with the small locality two. Some of these LRC codes are optimal LRC codes meeting the Cadambe-Mazumdar bound and some of them are very close to the Cadambe-Mazumdar bound. Several examples were given to illustrate our results. All examples were verified by Magma.\\

\end{document}